\documentclass[USenglish,a4paper,11pt]{scrartcl}

\usepackage[margin=1in]{geometry}
\usepackage[utf8]{inputenc}
\usepackage{lmodern}
\usepackage{hyperref}

\usepackage{amsmath}
\usepackage{amsthm}
\usepackage{thmtools}
\usepackage{amsfonts}
\usepackage{amssymb}
\usepackage{url}
\usepackage{tcolorbox}

\usepackage{algorithm}
\usepackage{algpseudocode}
\usepackage{algorithmicx}

\usepackage{forest}
\usepackage{relsize}
\usetikzlibrary{trees, snakes}
\usetikzlibrary{matrix}

\theoremstyle{plain}
\newtheorem*{theorem*}{Theorem}
\newtheorem{prop}{Proposition}
\newtheorem*{lemma*}{Lemma}
\newtheorem{theorem}{Theorem}
\newtheorem{corollary}{Corollary}
\newtheorem{lemma}{Lemma}

\newcommand{\norinf}[1]{\left\|#1\right\|_{\infty}}
\newcommand{\ZZ}{\mathbb{Z}}
\newcommand{\ZZgeq}{\mathbb{Z}_{\geq 0}}

\newcommand{\modsubadd}{\textsc{Mod-SubAdditivity Testing}}
\newcommand{\minconv}{\textsc{MinConv}}
\newcommand{\subsetsum}{\textsc{Unbounded SubsetSum}}
\newcommand{\frobenius}{\textsc{Frobenius}}
\newcommand{\allvalues}{\textsc{AllTargets}}
\newcommand{\subadd}{\textsc{SubAdditivity Testing}}

\title{On the Fine-Grained Complexity of the Unbounded SubsetSum and the Frobenius Problem}
\author{Kim-Manuel Klein\\
University of Kiel\\
\url{kmk@informatik.uni-kiel.de}}

\begin{document}

\date{}
\maketitle

\begin{abstract}
    Consider positive integral solutions $x \in \ZZgeq^{n+1}$ to the equation $a_0 x_0 + \ldots + a_n x_n = t$. In the so called unbounded subset sum problem, the objective is to decide whether such a solution exists, whereas in the Frobenius problem, the objective is to compute the largest $t$ such that there is no such solution.
    
    In this paper we study the algorithmic complexity of the unbounded subset sum, the Frobenius problem and a generalization of the problems. More precisely, we study pseudo-polynomial time algorithms with a running time that depends on the smallest number $a_0$ or respectively the largest number $a_n$. For the parameter $a_0$, we show that all considered problems are subquadratically equivalent to $(min,+)$-convolution, a fundamental algorithmic problem from the area of fine-grained complexity. By this equivalence, we obtain hardness results for the considered problems (based on the assumption that an algorithm with a subquadratic running time for $(min,+)$-convolution does not exist) as well as algorithms with improved running time. The proof for the equivalence makes use of structural properties of solutions, a technique that was developed in the area of integer programming. 
    
    In case of the complexity of the problems parameterized by $a_n$, we present improved algorithms. For example we give a quasi linear time algorithm for the Frobenius problem as well as a hardness result based on the strong exponential time hypothesis.
\end{abstract}

\section{Introduction}
Consider the integer program
\begin{align} \label{IP1}
    a_0 x_0 + \ldots + a_n x_n &= t\\
    x \in \ZZgeq^{n+1}&. \notag
\end{align}
for a given set of numbers $a_0, \ldots , a_n,t \in \ZZgeq$.
Having only a single constraint, (\ref{IP1}) is the most basic integer program in standard form. In this paper we consider fundamental algorithmic problems regarding (\ref{IP1}).
We call the numbers $a_0, \ldots , a_n$ the \emph{item sizes} and define the set of item sizes $S$ by $S = \{a_0, \ldots , a_n \}$. Throughout the paper we assume that the items are sorted by $a_0 < a_1 < \ldots < a_n$ and that $gcd(a_0, \ldots , a_n) = 1$ (this property is not really necessary as one can always divide by the gcd). We call $t$ the \emph{target value} and we say that $t$ is a feasible target value for given $a_0, \ldots, a_n$ if (\ref{IP1}) is feasible.

In the \subsetsum~problem the objective is to decide if (\ref{IP1}) is feasible, i.e. decide for given item sizes $S$ and target $t \in \ZZgeq$ if $t$ is feasible.
Formally it is defined by
\begin{tcolorbox}
    \subsetsum
    
    \textbf{Input:} Items $a_0, \ldots , a_n \in \ZZgeq$ and $t \in \ZZgeq$.\\
    \textbf{Objective:} Decide whether (1) is feasible.
\end{tcolorbox}
By a classical dynamic program that is being taught in undergraduate, the problem can be solved in time $O(nt)$ \cite{bellman1966dynamic_programming}.

In a related problem, the objective is to compute the largest $t$ such that (\ref{IP1}) is infeasible. This problem is intensively studied in the literature and is referred to by different names like the coin problem, the postage-stamp problem or the Frobenius problem (it was originally defined by Frobenius).
\begin{tcolorbox}
    \frobenius
    
    \textbf{Input:} Items $a_0, \ldots , a_n \in \ZZgeq$.\\
    \textbf{Objective:} Compute the largest $t \in \ZZgeq$ such that (\ref{IP1}) does not have a solution.
\end{tcolorbox}
The solution to the problem is called the \emph{Frobenius number} and we denote it by $F(a_0, \ldots , a_n)$. In the case that the item sizes $a_0, \ldots, a_n$ are clear from the context, we will only write $F$. For this problem, the property that $gcd(a_0, \ldots , a_n) = 1$ is essential, as $F(a_0, \ldots , a_n)$ would not be finite otherwise. It is easy to see that $F$ is finite if $gcd(a_0, \ldots , a_n) = 1$ and moreover $F \leq a_{n}^2$.

A natural generalization of both problems is a problem that computes the so called $\emph{residue table}$. The residue table was first considered by Brauer and Shockley~\cite{brauer1962residue_table} already in 1962 in order to solve \frobenius. In this problem, essentially, the feasibility for all target values $t$ is solved. We call the problem \allvalues.
\begin{tcolorbox}
    \allvalues
    
    \textbf{Input:} Items $a_0, \ldots , a_n \in \ZZgeq$.\\
    \textbf{Objective:} Compute $a[0], \ldots , a[n-1]$, such that  \begin{align*}
        a[i] = \min \{t \in \ZZgeq \mid t \equiv i \bmod a_0 \text{ and $t$ is a feasible solution of (\ref{IP1})}\}.
    \end{align*}
\end{tcolorbox}
Having a solution of \allvalues~at hand, one can decide in constant time if any given target $t$ is feasible (assuming that any element $a[i]$ in the sequence $a$ can be accessed in constant time). This is because
\begin{align*}
    t \text{ is feasible} \quad  \Leftrightarrow \quad t \geq a[t \bmod a_0]. 
\end{align*}
Note that in the case that $t > a[t \bmod a_0]$ a solution for target $t$ can be obtained by adding multiples of $a_0$ to the solution for target $a[t \bmod a_0]$ (which by definition exists). By the same argumentation, the Frobenius number can be determined from the residue table $a$ in $O(a_0)$ time. The Frobenius number $F$ is
\begin{align*}
    F = \max_i \{ a[i] \} - a_0. 
\end{align*}
Many algorithms in the literature that solve \subsetsum~or \frobenius~actually solve the more general problem of \allvalues~and then extract the respective solution for \frobenius~or \subsetsum~from the table. In this context one might ask if the \allvalues~problem is actually computationally equivalent to \subsetsum~and \frobenius~or if it is harder to compute the more general problem of determining the complete residue table. We have a partial answer on this question in this paper.

In general, \subsetsum~and \frobenius~are classical (weakly) NP-complete problems \cite{lueker1975_unbounded_subsetsum_NP-hardness,ramirez1996complexity_frobenius}, that is, if we assume that the given numbers $a_0,\ldots , a_n,t$ are encoded in binary. In this paper we study pseudo-polynomial time algorithms with a running time that depends on the smallest item size $a_0$ or the largest item size $a_n$.

\subsection{Related Results}
There are two different communities involved in studying the \subsetsum~and the \frobenius~problem.
On the one hand, there is the community that considers variants of subset sum problems.
On the other hand there is the community which studies the \frobenius~problem from various angles, which includes bounds for the Frobenius number, algorithmic complexity of the problem as well as practical fast algorithms.

\subsubsection*{Subset Sum}
One of the most fundamental problems in computed science is the (bounded) subset sum problem where the objective is to find a subset of a given (multi-)set of $n$ items such that there total sum equals $t$. This is equivalent to an integer program as in (\ref{IP1}), where additional upper bounds on the variables are given. \subsetsum~is in this sense a special case of the (bounded) subset sum problem, where the multiplicities of each item size are sufficiently large, i.e. $\geq t$.
Using a classical dynamic program by Bellman from 1957, the (bounded) subset sum problem can be solved in time $O(nt)$. This algorithm was improved by Koiliaris and Xu~\cite{koiliaris2019faster_subsetsum} who gave an algorithm with a running time of $\Tilde{O}(\sqrt{n}t)$ (omitting logarithmic factors) and then further by a breakthrough result of Bringmann~\cite{bringmann2017subset_sum_near_linear} who gave an algorithm with a running time of $\Tilde{O}(n + t)$. In that work, Bringmann also developed an algorithm for \subsetsum~with a running time of $O(t \log t)$.
Jansen and Rohwedder~\cite{Jansen_rohwedder_IP_convolution} showed later that there is an algorithm with a running time $O(a_n \log^2 a_n)$ for \subsetsum~for the stronger parameter $a_n$. Assuming the strong exponential time hypothesis (SETH), it was shown that the existing algorithms are actually optimal as Abboud et al.~\cite{abboud_bringmann_hardness_subset_sum_bringmann2019} proved that there is no algorithm for subset sum or \subsetsum~with a running time of $O(t^{1-\epsilon})$ for any $\epsilon>0$. 

The complexity of \subsetsum~regarding the parameter $a_0$ was open so far. In the hard instance that was used in \cite{abboud_bringmann_hardness_subset_sum_bringmann2019} all items $a_0, \ldots , a_n$ and $t$ differ only by a factor of $2$ and therefore only an algorithm with a running time of $O(a_{0}^{1-\epsilon})$ can be excluded.

\subsubsection*{The Frobenius Problem}
The \frobenius~problem is an intensively studied problem in the literature. Since we can not cover all its aspects here, we refer to the monograph of Ramirez-Alfonsin~\cite{alfonsin2005frobenius_book} for a general overview. The book covers over 400 sources that are concerned with the problem.

In the case that $n=2$, the Frobenius number $F$ can be easily determined by the formula $F = a_0 a_1 - a_0 - a_1$. However for $n>2$ there is no formula but several upper and lower bounds are known for $F$ (see for example \cite{aliev2007frobenius_lower_bound,fukshansky2007frobenius_upper_bound}). A bound that we make use of in this paper is due to Erdös and Graham~\cite{erdos1972frobenius_bound} which states that
\begin{align*}
    F(a_0, \ldots , a_n) \leq \frac{2 a_{n-1} a_n}{n+1} - a_n =  O\big(\frac{a_{n}^2}{n}\big) .
\end{align*}
As mentioned, the \frobenius~problem is NP-complete and hence we can not expect a polynomial time algorithm. However, many pseudo-polynomial time algorithms have been studied having a running time which depends on the input parameters.
The most relevant result for our paper here is the result of Böcker and Lipt{\'a}k~\cite{bocker2007_frobenius_algorithm}. They developed an algorithm solving \allvalues~and therefore also \frobenius~and \subsetsum~with a running time of $O(n a_0)$. Based on the Round Robin procedure, they presemt a very elegant dynamic program computing the residue table. By this, Böcker and Lipt{\'a}k improved upon several known existing pseudo-polynomial time algorithms for the \frobenius~problem (see Chapter 1 in \cite{alfonsin2005frobenius_book}).
The algorithm of Nijenhuis~\cite{nijenhuis1979frobenius_dijkstra} for example has a running time of $O(a_0(n + \log a_0))$. Nijenhuis used Dijkstras algorithm in the weighted Cayley graph of the group $(\ZZ_{a_0},+)$ to compute the residue table. The performance of several pseudo-polynomial time algorithms for the problem has been tested in practice~\cite{beihoffer2005practical_frobenius_nijenhius}.

The mentioned pseudo-polynomial time algorithms are efficient in the case that the input numbers are small. There are also algorithms for \frobenius~in the case that $n$ is small. For example, for three coins, Greenberg~\cite{greenberg1988frobenius_three_items} developed an algorithm with a running time of $O(\log a_0)$. Kannan~\cite{kannan1992frobenius_polynomiality_constant_n} showed in the end that there is a polynomial time algorithm for any constant $n$. Other practically efficient algorithms have been developed that are efficient if $n$ is small (see for example \cite{einstein2007frobenius, roune2008solving_frobenius_small_n}).

\subsection{Our Results:}
In this paper we prove several algorithmic results as well as hardness results based the above mentioned assumption for $(min,+)$-convolution (which we denote by \minconv) and the SETH.
Concerning parameterization by $a_0$, we prove the following theorem, which shows that all our considered problems are subquadratically equivalent.
\begin{theorem}\label{thm:equivalence}
    The following statements are equivalent:
    \begin{enumerate}
        \item \minconv~can be solved in time $O(n^{2-\epsilon})$.
        \item \subsetsum~can be solved in time $O(a_{0}^{2-\epsilon})$.
        \item The \frobenius~problem can be solved in time $O(a_{0}^{2-\epsilon})$.
        \item \allvalues~can be solved in time $O(a_{0}^{2-\epsilon})$.
    \end{enumerate}
\end{theorem}
More precisely, we show in Theorem~\ref{thm:subsetsum_by_conv} that \allvalues~can be solved by using only at most $\lceil \log \log a_0 \rceil$ times $(min,+)$-convolution as a subroutine. Using the algorithm of Williams~\cite{williams2018faster_apsp}, we obtain a faster algorithm for the considered problems with a running time of $\frac{a_{0}^2}{2^{\Omega( \log a_0)^{1/2}}}$ beating the best known algorithm of Böcker and Lipt{\'a}k~\cite{bocker2007_frobenius_algorithm} in the case that $n \in \Omega(\frac{a_0}{\log^c a_0})$ for any constant $c$. 
The theorem relies on a structural result for solutions of (\ref{IP1}) that we show in Lemma~\ref{lem:structure_main} which is of interest by itself. This type of structural result was first used to bound the number of non-zero components of feasibility IPs.



Concerning parameterization by the largest item size $a_n$, we present the following results:
\begin{itemize}
    \item Algorithm~\ref{alg:subsetsum} solving \subsetsum~with a running time of $O(a_n \log a_n \cdot \log (\frac{F}{a_n}))$ improving upon the algorithm by Jansen and Rohwedder~\cite{Jansen_rohwedder_IP_convolution} in the case that the Frobenius number $F$ is small, i.e. $F << a_{n}^2$. For example, for average case instances the expected Frobenius number is rather small~(see \cite{aliev2011expected_frobenius}).
    \item Algorithm~\ref{alg:frobenius} solving \frobenius~with a running time of $O(a_n \log a_n \cdot \log (\frac{F}{a_n}))$. To our best knowledge this is the first algorithm solving the problem in quasi linear time by an input parameter. It beats the algorithm obtained by Theorem~\ref{thm:equivalence} in the case that $a_n << a_{0}^2$.
    \item Algorithm~\ref{alg:allvalues_an} solving \allvalues~with a running time of $O(a_{n}^{3/2} \log^{1/2} a_n)$. To our best knowledge the paramterized complexity of the problem has not been considered yet with respect to the parameter $a_n$.
    \item Based on the existing hardness result by Abbout et al.~\cite{abboud_bringmann_hardness_subset_sum_bringmann2019}, Theorem~\ref{thm:hardness_frobenius_an} proves hardness for \frobenius~for the parameter $a_n$ and therefore shows that Algorithm~\ref{alg:frobenius} is nearly optimal (assuming SETH).
\end{itemize}

\subsection{Hardness Assumptions}
A fundamental problem in the area of fine grained complexity is $(min,+)$-convolution which we denote by \minconv.
\begin{tcolorbox}
    \minconv
    
    \textbf{Input:} Sequences $(a[i])_{i=0}^{n-1}$ and $(b[i])_{i=0}^{n-1}$.\\
    \textbf{Objective:} Compute a sequence $(c[i])_{i=0}^{n-1}$, such that \begin{align*}
        c[k] = \min_{i+j=k} (a[i]+ b[j]).
    \end{align*}
\end{tcolorbox}
It is assumed that there does not exist a truly subquadratic algorithm for the problem.

\textbf{Assumption:} \minconv~can not be solved in time $O(n^{2-\epsilon})$ for any $\epsilon>0$.

This assumption is the base of several lower bounds in fine-grained complexity (\cite{backurs2017minconv_usage,cygan2019_minconv,kunnemann2017minconv_usage}).
Cygan et al.\cite{cygan2019_minconv} showed equivalence of \minconv~to several problems including the unbounded knapsack problem parameterized by $a_n$ - a problem related to \subsetsum.

One of the most basic hardness assumption that we make use of is based on the $k$-SAT problem, where the objective is to decide the feasibility of a boolean formula in conjunctive normal form having at most $k$ variables per clause. The strong exponential time hypothesis (SETH) was stated by Impagliazzo and Paturi~\cite{impagliazzo2001SETH} and is has been used since for dozens of problems in Parameterized Complexity and for problems within P to show lower bounds.

\textbf{Assumption (SETH):} There does not exist an algorithm for $k$-SAT with a running time of $O(2^{(1-\epsilon)n})$ for any $\epsilon > 0$.

\section{Complexity of Parameterization by $a_0$}
In this section we prove Theorem~\ref{thm:equivalence} regarding the equivalence of \minconv~to \allvalues, \subsetsum~and the \frobenius~problem when parameterized by the smallest item size $a_0$.
\begin{center}
    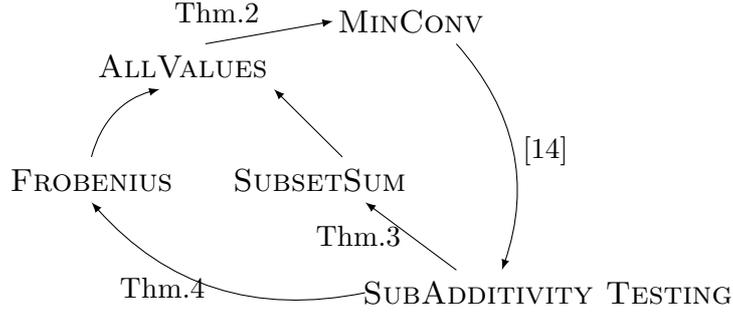
\begin{figure}
    \begin{tikzpicture}[scale=0.6]
\usetikzlibrary{positioning,arrows}
\tikzset{
  state/.style={circle,draw,minimum size=6ex},
  arrow/.style={-latex, bend angle=45}}

\node[] at (5,0) {\large \textsc{MinConv}};
\draw[arrow] (0.5,-0.5) -- node[anchor=south east] {Thm.\ref{thm:subsetsum_by_conv}} (3.3,0);
\node[] at (0,-1) {\large \textsc{AllValues}};
\node[] at (-2,-3.5) {\large \textsc{Frobenius}};
\draw[arrow, bend right] (-2,-3) to [out=30,in=150] node[anchor=east] {} (-0.5,-1.5);
\draw[arrow, bend right] (4,-6) to [out=30,in=150] node[anchor=east] {Thm.\ref{thm:hardness_frobenius}} (-2,-4);
\node[] at (3,-3.5) {\large \textsc{SubsetSum}};
\draw[arrow] (3.5,-3) to node[anchor=east] {} (2,-1.5);
\draw[arrow] (6,-5.5) to node[anchor=east] {Thm.\ref{thm:hardness_subsetsum}} (4,-4);
\node[] at (8,-6) {\large \textsc{SubAdditivity Testing}};
\draw[arrow, bend right] (6,-0.5) to [out=30,in=150] node[anchor=west] {\cite{cygan2019_minconv}} (7,-5.5);
\end{tikzpicture}
    \centering
    \caption{Overview of the reductions. An arrow from $A$ to $B$ denotes a reduction from problem $A$ to problem $B$.}
    \label{fig:overview}
\end{figure}

\end{center}

\subsection{Structural Properties of Solutions}
Before we show Theorem~\ref{thm:equivalence}, we prove a structural result regarding the existence of solutions of~(\ref{IP1}) with a specific shape. The shown structural property will then be of use in Theorem~\ref{thm:subsetsum_by_conv}.

The proof of the structural result relies on a technique that was developed in the area of integer programming showing the existence of solutions of bounded support~\cite{eisenbrand2006_support}. The support $supp(x)$ of a solution $x$ is the set of non-zero components.
\begin{lemma}\label{lem:structure_main}
    Assuming that (\ref{IP1}) is feasible, then there exists a solution $x$ of (\ref{IP1}) with \begin{align*}
        \prod_{i=1}^n (x_i+1) \leq a_0.
    \end{align*}
\end{lemma}
\begin{proof}
    Let $x \in \ZZgeq^{n}$ be a solution of (\ref{IP1}) that is lexicographically maximal. We will show that for this solution $\prod_{i=1}^n (x_i+1) \leq a_0$ holds.
    
    Suppose by contradiction that $\prod_{i=1}^n (x_i+1) > a_0$. By the pigeonhole principle there exist two distinct subvectors $x', x'' \leq x$ with
    \begin{align*}
        a^T x' \equiv a^T x'' \bmod a_0
    \end{align*}
    and hence $a^Tx' = a^T x'' + ka_0$ for some $k \in \ZZ$. Without loss of generality we assume that $k \geq 0$ (otherwise we would exchange $x'$ and $x''$). If $k = 0$ we assume that $x'$ is lexicographically smaller than $x''$.\\
    The solution $x$ can now be modified to obtain a solution $y \in \ZZgeq^{n}$ by 
    \begin{align*}
        y = x + x'' - x' + k e_0.
    \end{align*}
    The solution $y$ is feasible since 
    \begin{align*}
       t = a^T x = a^T x + a^T (x'' - x' + k e_0) = a^T y,
    \end{align*}
    as by definition 
    \begin{align*}
        a^T (x''-x'+ k e_0) = a^T (x''-x') + k a_0 = 0,
    \end{align*}
    holds, where $e_0$ is the $0$-th unit vector. 
    Furthermore $y \geq 0$ holds as $x - x' \geq 0$ and therefore we can conclude that $y$ is a feasible solution of IP~(\ref{IP1}).
    
    If $k>0$ then $y_0 = x_0 + k$ and hence the solution $y$ is lexicographically larger than $x$, which is a contradiction to the assumption.
    If $k = 0$ and hence $x'$ is lexicographically smaller than $x''$ there exists a component $i \in (supp(x') \cup supp(x''))$ such that $x''_{i} > x'_{i}$ and $x'_{j} = x''_{j}$ holds for all $j<i$. However, this implies again that $y$ is lexicographically larger than $x$ which is a contradiction to the assumption. Therefore $\prod_{i=1}^n (x_i+1) \leq a_0$ holds.
    \end{proof}
The following corollary shows the existence of a solution of (\ref{IP1}) with bounded support. It follows directly from the previous lemma. The same statement was previously already shown by Aliev et al.~\cite{aliev2020optimizing_sparsity} using a different proof technique. In their proof they use a geometric volume argument in combination with Minkowski's first theorem from the geometry of numbers. Before that Aliev et al.\cite{aliev2017_spars_diophantine} had also shown a weaker bound of $\log(a_n) +1$.
\begin{corollary}
    There exists a solution $x$ of (\ref{IP1}) with \begin{align*}
        |supp(x)| \leq \log(a_0) +1.
    \end{align*}
\end{corollary}
\begin{proof}
    By Lemma~\ref{lem:structure_main} there exists a solution $x$ of (\ref{IP1}) with 
    \begin{align*}
        2^{|supp(x)  \setminus \{0 \}|} \leq \prod_{i=1}^n (x_i+1) \leq a_0.
    \end{align*}
    We obtain that $|supp(x)| \leq \log(a_0) +1$.
\end{proof}

\subsection{Solving \allvalues~by \minconv}
In order to solve \allvalues~we use a min-convolution operation, where the indices of the sequences are computed $\bmod n$.
Let the \emph{modulo-min-convolution} $a \oplus_{n} b$ of two sequence $(a[i])_{i=0}^{n-1}$ and $(b[i])_{i=0}^{n-1}$ be defined by 
\begin{align*}
    (a \oplus_n b)[k] = \min_{i+j \equiv k \bmod n} a[i] + a[j].
\end{align*}
The modulo-min-convolution can easily be computed using the classical convolution on modified sequences $\bar{a}$ and $\bar{b}$ of length $2n$ defined by 
\begin{align*}
    \bar{a}[i] = a[i \bmod n] \quad \text{ for } i = 0,\ldots , 2n-1\\
    \bar{b}[i] = a[i \bmod n] \quad \text{ for } i = 0,\ldots , 2n-1.
\end{align*}
We compute the sequence $\bar{c}$ using classical min-convolution, i.e.
\begin{align*}
    \bar{c} = \bar{a} \oplus \bar{b}.
\end{align*}
The sequence $c = a \oplus_n b$ is then obtained by
\begin{align*}
    c[i] = \min \{\bar{c}[i],\bar{c}[i+n]\} \quad \text{ for } i = 0,\ldots , n-1.
\end{align*}
Applying the mod-min-convolution repeatedly to a sequence $(a[i])_{i=0}^{n-1}$ yields a sequence $a^*$ with
\begin{align*}
    a^* = a \oplus_n a \oplus_n  \cdots,
\end{align*} where $a^* = a^* \oplus_n a$ holds. For this sequence there exist no indices $i,j$ with $a^*[i] + a^*[j] < a^*[i+j \bmod n]$. In this sense $a^*$ is the \emph{transitive mod-min-convolution closure} of the sequence $a$. 
It is easy to see that a solution to the \allvalues~problem is actually the transitive mod-min-convolution closure of the sequence $(a[i])_{i=0}^{a_0-1}$ defined by
\begin{align*}
    a[i] = \min \{ a_j \mid a_j \equiv i \bmod a_0  \}.
\end{align*}
In the following theorem we show that $a^*$ can be computed by using only $\lceil \log \log a_0 \rceil$ many mod-min-convolution operations.

\begin{theorem}\label{thm:subsetsum_by_conv}
    An algorithm for \minconv~with a running time of $T(n)$ implies an algorithm with a running time of 
    \begin{align*}
        \lceil \log \log a_0 \rceil \cdot T(a_0)  + O(a_0 \log a_0)
    \end{align*} for \allvalues.
\end{theorem}
\begin{proof}
    Define the set $V$, which consists of the set of item sizes $S$ and multiples of $a_i \in S$ which are powers of $2$, i.e.
    \begin{align*}
        V = \bigcup_{1 \leq i \leq n} \left( \bigcup_{0 \leq j \leq \log(a_0)} 2^j a_i \right)
    \end{align*}
    Obviously, every element of $V$ is a feasible target value and also every sum of elements of $V$ is a feasible target value.
    
    Define the sequence $A$ of length $a_0$ by
    \begin{align*}
        A[i] = \min \{ t \in V \mid t \equiv i \bmod a_0 \}.
    \end{align*}
    The sequence $A$ can be computed in $O(a_0 \log a_0)$ time in a straightforward way. Starting with $a_1$, each item $a_i$ is iteratively multiplied by $2$ and the value $2^j a_i$ is stored in $A[2^j a_i \bmod a_0]$ if it is smaller than its current content (the sequence is initialized by infinity values).
    
    Compute the sequence \begin{align*}
        A^{*} = \underbrace{A \oplus_{a_0} \ldots \oplus_{a_0} A}_{\lfloor \log a_0\rfloor}
    \end{align*}
    which is obtained by a $\lfloor \log a_0 \rfloor$-fold mod-convolution of $A$.
    
    \textbf{Claim:} The sequence $A^{*}$ is the solution to the \allvalues~problem.
    
    Clearly, each element $A^{*}[i]$ is a feasible target value with $A^{*}[i] \equiv i \bmod a_0$ as it is contained in the sum of elements of $V$.
    
    On the other hand, let $t_i$ be the minimum feasible target value with $t_i \equiv i \bmod a_0$.
    By Lemma~\ref{lem:structure_main} we know that for each $t_i \in \ZZgeq$ there exists a solution $x \in \ZZgeq^{n}$ with $\sum_{i=1}^n \log (x_i+1) \leq \log a_0$ (note that $x_0 = 0$ otherwise, $t_i$ would not be minimal). This implies that $t_i = x_1 a_1 + \ldots + a_n x_n$ can be written by the sum of 
    \begin{align*}
        \sum_{i=1}^n \log (x_i +1) \leq \log a_0
    \end{align*}
    many elements of $V$ (by binary encoding, each multiplicity $x_i$ of an element $a_i$ is contained in the sum of at most $\log (x_i +1)$ multiplicity elements $2^j a_i$ contained in $V$). Therefore, every $t_i$ is contained in the $\lfloor \log a_0 \rfloor$-fold addition of $V$, i.e. $t_i \in \underbrace{V \oplus_{a_0} \ldots \oplus_{a_0} V}_{\lfloor \log a_0\rfloor}$ which is equivalent to $\lfloor \log a_0 \rfloor$-fold Min-Convolution of $A$ and hence $t_i = A^{*}[i]$ which proves the claim.
    
    By using a binary exponentiation argument we can reduce the number of calls to the mod-convolution procedure further and compute $A^{*}$ more efficiently. Therefore, define
    \begin{align*}
        B^{(i)} = \underbrace{A \oplus_{a_0} \ldots \oplus_{a_0} A}_{2^i}.
    \end{align*}
    The sequences $B^{(i)}$ can be computed efficiently using $i$ calls to the mod-convolution procedure by $B^{(1)} = A$ and $B^{(i)} = B^{(i-1)} \oplus_{a_0} B^{(i-1)}$ for $i > 1$. Finally, $A^{*} = B^{r}$ is obtained for $r \geq \lceil \log \log a_0 \rceil$. Note that as the entries of $A$ are minimal, we have that $A^{*} \oplus_{a_0} A = A^{*}$ and therefore it does not hurt to apply too many mod-convolutions on $A$.
\end{proof}
    Using the famous algorithm of Williams~\cite{williams2018faster_apsp} to solve ~\minconv~in time $\frac{n^2}{2^{\Omega( \log n)^{1/2}}}$ yields improved algorithms for the considered problems. By this we improve upon the algorithm of Böcker and Lipt{\'a}k~\cite{bocker2007_frobenius_algorithm} in the case that the number of items $n$ is not small, i.e. $n \in \Omega(\frac{n}{\log^c n})$ for any constant $c$ (note that the term $2^{\Omega( \log n)^{1/2}}$ grows asymptotically faster than any polylogarithmic function).
\begin{corollary}
    The \subsetsum, the \frobenius~and the \allvalues~problem can be solved with a running time of
    \begin{align*}
        \frac{a_{0}^2}{2^{\Omega( \log a_0)^{1/2}}}.
    \end{align*}
\end{corollary}
Furthermore, it is known that \minconv~can be solved for random instances in near linear expected time~\cite{bussieck1994min_conv_random_instances}. Using this algorithm as a subroutine might yield a very efficient practical algorithm that solves \allvalues.

\subsection{Hardness of \subsetsum~and \frobenius}
In the following we show hardness of the \subsetsum~and the \frobenius~problem by a reduction from \subadd, a problem that is known to be subquadratically equivalent to \minconv \cite{cygan2019_minconv}. The problem of \subadd~is to decide for a given sequence $(a[i])_{i=0}^{n-1}$ if there exist indices $i,j$ such that $a[i] + a[j] < a[i+j \bmod n]$. We consider the modulo version of the problem which we call \modsubadd.
\begin{tcolorbox}
    \modsubadd
    
    \textbf{Input:} Sequence $(a[i])_{i=0}^{n-1}$\\
    \textbf{Objective:} Decide whether there exist indices $i,j$ such that $a[i] + a[j] < a[i+j \bmod n]$.
\end{tcolorbox}
In other words, the problem is to decide if a given sequence $a$ is idempotent under a mod-min-convolution operation, i.e. if $a \oplus_n a = a$.

It is easy to see that \modsubadd~is subquadratically equivalent to \subadd. 
\begin{prop}
    If there exists an algorithm with a running time of $T(n)$ for \subadd~then there exists an algorithm with a running time of $T(2n) + O(n)$ for \modsubadd~and vice versa.
\end{prop}
\begin{proof}
    Given is a sequence $(a[i])_{i=0}^{n-1}$ in the \modsubadd~problem. Construct a sequence $(b[i])_{i=0}^{2n-1}$ of double length defined by
    \begin{align*}
        b[i] = a[i \bmod n] \quad \text{ for } i = 1,\ldots , 2n
    \end{align*}
    There exist indices $i,j$ such that $b[i] +b[j] < b[i+j]$ if and only if there exist indices $i,j \leq n$ with $a[i] + a[j] < a[i+j \bmod n]$. Therefore, an algorithm of running time $T(n)$ for the \subadd~problem implies an algorithm of running time $T(n) + O(n)$ for the \modsubadd~ problem.

    On the other hand, let $(a[i])_{i=1}^n$ be a sequence in the \subadd~problem. Construct a sequence $(b[i])_{i=1}^{2n}$ of double length defined by
    \begin{align*}
        b[i] = \begin{cases} a[i] & \text{ for $i \leq n$}\\ - \infty & \text{ otherwise.}\end{cases}
    \end{align*}
    There exist indices $i,j$ such that $b[i] +b[j] < b[i+j]$ if and only if there exist indices $i,j \leq n$ with $a[i] + a[j] < a[i+j]$. Therefore, an algorithm of running time $T(n)$ for the \modsubadd~problem implies an algorithm of running time $T(n) + O(n)$ for the \subadd~problem.
\end{proof}
In the following theorem we show the reduction of \modsubadd~(respectively \subadd) to \subsetsum.
\begin{theorem}\label{thm:hardness_subsetsum}
    An algorithm for \subsetsum~with a running time of $T(a_0)$ implies an algorithm with a running time of $T(n) + O(n)$ for \subadd.
\end{theorem}
\begin{proof}
    Given is a sequence $a[1], \ldots , a[n-1]$. Since $a[0] + a[i] \geq a[i+0]$ the entry $a[0]$ does not need to be considered as it can not be used to violate the modular-subadditivity property of a sequence.
    
    The items of the \subsetsum~instance are defined by 
    \begin{align}\label{instance:subsetsum}
        a_0 &= 2n\\ \notag
        a_i &= a_0 a[i] + 2i \qquad \qquad \qquad \text{for $i = 1, \ldots , n-1$} \\ \notag
        \bar{a}_i &= a_0(M-a[i]) - 2i -1 \qquad \text{for $i = 1, \ldots , n-1$}, \notag
    \end{align}
    where $M$ is defined by $M = 2 ( \max_{i \in \{1, \ldots ,n-1\}} a[i] + n)$.
    We define the target value $t$ by 
    \begin{align*}
        t = a_0(M-1) -1.
    \end{align*}
    
    Suppose that the sequence $a$ is not modulo-subadditive. Then we can build a solution of (\ref{IP1}) in the following way:
    Given are indices $i,j$ with $a[i] + a[j] < a[i+j \bmod a_0]$. Summing up the item sizes of $\bar{a}_{(i+j \mod p)}, a_i$ and $a_j$ we obtain that
    \begin{align*}
        & \quad \bar{a}_{(i+j \bmod p)} + a[i] + a[j]\\
        &= \,a_0 (M-a[i+j \bmod a_0]) - 2(i+j) -1 + a_0 a[i] + 2i + a_0 a[j] + 2j\\
        & = \,a_0 (M - a[i+j \bmod a_0] + a[i] + a[j]) - 1 \\
        & = \, a_0 (M-r) - 1,
    \end{align*}
    for some $r \geq 1$. Hence $\bar{a}_{(i+j \mod p)} + a_i + a_j + ra_0 = t$ which proves feasibility of (\ref{IP1}).
    
    In the following we show that the converse is also true, i.e. if there exists a solution of (\ref{IP1}) then the sequence $a$ is not modulo-subaddtive.
    
    \textbf{Observation 1:} A feasible solution of $(\ref{IP1})$ has to contain exactly one item $\bar{a}_i$.\\
    Since $t<M$ and every item $\bar{a}_i > M/2$, a feasible solution can contain at most one item $\bar{a}_i$. On the other hand, since every item $a_i$ is even and $a_0$ is also even, a sum with indices $I \subset [n-1]$ of items fulfills $\sum_{i \in I} a_i \equiv 2y \bmod a_0$ for some $y\in \ZZ_{a_0}$. This implies that $\sum_{i \in I} a_i \not\equiv t \equiv -1 \bmod a_0$. 
    Therefore, a feasible solution has to contain at least one item $\bar{a}_i$.
    
    \textbf{Observation 2:} Consider a feasible solution of (\ref{IP1}) using one item $\bar{a}_p$ and items $\{a_i \mid i \in I\}$ for some $I \subseteq [n-1]$. Then $p \equiv \sum_{i \in I} i \bmod a_0$ holds.\\
    Because of the feasibility, we know that $\bar{a}_p + \sum_{i \in I} i \equiv t \bmod a_0$, which implies
    \begin{align*}
        \bar{a}_p + \sum_{i \in I} i &\equiv t \bmod a_0 \\
        \Leftrightarrow \quad -2p -1 + \sum_{i \in I} 2i &\equiv -1 \bmod a_0 \\
        \Leftrightarrow \quad \sum_{i \in I} 2i &\equiv 2p \bmod a_0\\
        \Leftrightarrow \quad \sum_{i \in I} i &\equiv p \bmod a_0
    \end{align*}
    
    
    Using the observations above, we can now show that there exists a feasible solution of (\ref{IP1}) only if the sequence $a[1], \ldots , a[n-1]$ is not modulo-subadditive.
    
    Suppose that there exists a feasible solution of (\ref{IP1}). Then by Observation~1 and 2 the solution uses a single item $\bar{a}_p$ and a set of items $\{a_i \mid i \in I \}$ for some $I \subseteq [n-1]$ with $p \equiv \sum_{i \in I} i \bmod a_0$.
    The solution is feasible only if 
    \begin{align*}
        \bar{a}_p + \sum_{i \in I} a_i & \leq t\\
        \Leftrightarrow \quad a_0 (M-\bar{a}[p] + \sum_{i \in I} a[i]) - 1 & \leq a_0(M-1)-1 \\
        \Leftrightarrow \quad -a[p] + \sum_{i \in I} a[i] &\leq -1 \\
        \Leftrightarrow \quad \sum_{i \in I} a[i] &<  a[p]
    \end{align*}
    This implies that there exist indices $i_1, \ldots, i_{|I|}$ with $i_1+ \ldots + i_{|I|} \equiv p \bmod a_0$ such that $a[i_1]+ \ldots + a[i_l] < a[p]$. Note that $I$ can not be the empty set as there is no item $\bar{a}_p$ with $\bar{a}_p \equiv t \bmod a_0$. Also $|I| \neq 1$ holds as
    \begin{align*}
        \bar{a}_p + a_p = a_0 M -1 > t.
    \end{align*}
    For $|I| = 2$ we know there exist indices $i,j \in I$ with $a[i] + a[j] < a[i+j]$ and therefore the sequence $a[1], \ldots a[n-1]$ is not mod-subadditive.
    In the case of $|I|>2$ we can set $i = i_1$ and $j = i_2 + \ldots + i_{|I|} \bmod a_0$. Hence, either $a[j] < a[i_2] + \ldots + a[i_{|I|}]$ holds and we can proceed inductively with the smaller index set $I' = \{i_2, \ldots , i_{|I|} \}$ or 
    \begin{align*}
        a[i+j \bmod a_0] >  a[i_1] + a[i_2] + \ldots + a[i_{|I|}] \geq a[i] + a[j]
    \end{align*}
    holds and hence mod-subadditivity is violated for indices $i$ and $j$.
\end{proof}
Taking a close look at instance~(\ref{instance:subsetsum}) yields the same statement for \frobenius.
\begin{theorem}\label{thm:hardness_frobenius}
    An algorithm for the \frobenius~problem with a running time of $T(a_0)$ implies an algorithm with a running time of $T(n) + O(n)$ for \subadd.
\end{theorem}
\begin{proof}
    Given is the sequence $a[0], \ldots , a[n-1]$. We define the instance of item sizes for the \frobenius~problem just as in the previous Theorem~\ref{thm:hardness_subsetsum}:
    \begin{align*}
        a_0 &= 2n\\ \notag
        a_i &= a_0 a[i] + 2i \qquad \qquad \qquad \text{for $i = 1, \ldots , n-1$} \\ \notag
        \bar{a}_i &= a_0(M-a[i]) - 2i -1 \qquad \text{for $i = 1, \ldots , n-1$}, \notag
    \end{align*}
    We will show that the frobenius number of this instance is $< a_0 M -1$ if and only if the sequence $a$ is not mod-subadditive.
    
    Observe in the table below that for every target value $t \not\equiv -1 \bmod a_0$ there is an item $a_i$ or $\bar{a}_i$ such that $t \equiv a_i \bmod a_0$ (if $t$ is even) or $t \equiv \bar{a}_{i} \bmod a_0$ (if $t$ is odd).
    \begin{center}
        \begin{tabular}{|c|c|c|c|c|c|c|c|c|c|c|c|c|}
        \hline
         Item\, $\mod a_0$& $0$ & $1$ & $2$& $3$ & $4$ & \phantom{$a_1$} & $\cdots$ & \phantom{$a_1$} &  $2n-4$ & $2n-3$ & $2n-2$ & $2n-1$ \\ \hline
         Item & $a_0$ & $\bar{a}_{n-1}$ & $a_1$ & $\bar{a}_{n-2}$ & $a_2$ & & $\cdots$ & &$a_{n-2}$ & $\bar{a}_{1}$ & $a_{n-1}$ &\\
         \hline
    \end{tabular}
    \end{center}
    Since the item sizes $a_i$ and $\bar{a}_i$ are $< a_0 M -1$ every instance with a target value $t \not\equiv -1 \bmod a_0$ has a solution $< a_0 M -1$.
    Consider a target value $t \equiv -1 \bmod a_0$: According to the proof of the previous theorem there exists a solution for $t = a_0 (M-1) -1$ if and only if the sequence $a$ is not mod-subadditive.
    This implies that the frobenius number of the instance is $< a_0 M -1$ if and only if the sequence $a$ is not mod-subadditive.
\end{proof}

\section{The Complexity of Parameterization by $a_n$}
In this section, we present algorithms for $\subsetsum$, \frobenius~and \allvalues~parameterized by the largest item size $a_n$. An algorithm for \subsetsum~with a running time of $O(a_n \log^2 a_n)$ was already known. The algorithm by Jansen and Rohwedder~\cite{Jansen_rohwedder_IP_convolution} solves general feasibility IPs for given matrix $A \in \ZZ^{m \times n}$ and vector $b \in \ZZ^m$ of the form \begin{align*}
    Ax = b\\ x \in \ZZgeq^n
\end{align*}
with a running time of $O((\sqrt{m} \norinf{A})^m \log \norinf{A} \cdot\log (\norinf{A} + \norinf{b})) + O(nm)$. Using this algorithm for $m=1$ in combination with the proximity result of Eisenbrand and Weismantel~\cite{eisenbrand2018proximity} in order to reduce $t$ (i.e. $t \leq a_{n}^2$) leads to the mentioned running time for solving \subsetsum.

In Section~\ref{sec:alg_subsetsum_an} we describe a modified and simplified version of the algorithm by Jansen and Rohwedder~\cite{Jansen_rohwedder_IP_convolution} for solving \subsetsum~with a running time of $O(a_{n} \log^2 a_{n})$. In contrast to their algorithm, our algorithm computes all feasible target values for a certain fixed set of intervals of length $a_n$. This has the advantage that our algorithm is able to detect if the Frobenius number of the instance is exceeded and terminates. By this we do not need to use the integer programming proximity to reduce the target value $t$ and instead use a bound on the Frobenius number which is in general better. Therefore, our modified algorithm improves upon the running time if $F<a_{n}^2$. However, the main reason why we present this version of the algorithm is that we can show in Section~\ref{sec:alg_frobenius}, that our modified algorithm can be easily adapted to compute a solution of the \frobenius~problem.

In the algorithm, we compute the \emph{sumset} $S+T$ of two sets $S,T \subset \ZZ$ defined by
\begin{align*}
    S + T = \{s+t \mid s \in S, t \in T \}.
\end{align*}
For two sets $S \subseteq I \cap \ZZ$ and $T \subseteq I' \cap \ZZ$ that are contained in intervals $I,I'$ of size $n$, it is well known that the sumset $S+T$ can be computed in time $O(n \log n)$. This can done by a reduction to boolean convolution which can be solved by an FFT approach (see also \cite{Jansen_rohwedder_IP_convolution, bringmann2017subset_sum_near_linear}). 

Define the intervals 
\begin{align*}
    I^{(j)} = ((j-1)a_{n}, j a_{n}]
\end{align*} 
and let $A^{(j)}$ be the set of all feasible target values within the interval $I^{(j)}$, i.e.
\begin{align*}
    A^{(j)} = \{t \in I^{(j)} \mid \text{where (\ref{IP1}) is feasible} \}.
\end{align*}
It is easy to see that
\begin{align*}
    |A^{(1)}| < \ldots < |A^{(f)}| < |A^{(f+1)}| = n
\end{align*}
holds, where $f$ is the largest number such that $|A^{(f)}| <n$.
Each target value in the interval $I^{(f+1)}$ is feasible, which implies that each target value in $|A^{(k)}|$ for $k>f$ has to be feasible as well.

In the following lemma we state a central property of the feasibility sets $A^{(i)}$ that will be essential in the following algorithm as it allows to compute feasibility sets from other feasibility sets with smaller indices. A similar recursive argument was used in \cite{Jansen_rohwedder_IP_convolution}.
\begin{lemma}\label{lem:feasibility_sets_addition}
    For every $i,j \geq 1$ the following equation holds:
    \begin{align*}
        A^{(i+j)} = (A^{(i)} + A^{(j)} + S) \cap I^{(i+j)}
    \end{align*}
\end{lemma}
\begin{proof}
    Let $t \in A^{(i+j)}$ be a feasible target value. By definition, there exists a solution $x \in \ZZ^n$ such that $\sum_{i=0}^{n} x_i a_i = t \leq (i+j)a_n$.
    There exist subvectors $x', x'' \leq x$ such that $x'+x'' + e_k = x$ for some unit vector $e_k$ and 
    \begin{align*}
    t' := \sum_{l=0}^{n} x'_{l} a_l <&~ i a_n,\\
    t'' := \sum_{k=l}^{n} x''_{l} a_l <&~ j a_n,\\
   t' + t'' + a_k = &~  t,
    \end{align*}
    holds. Hence there exists a split of this sum into $t' + t'' + a_k = t$ for some $a_k \in S$, $t' \in A^{(i)}$ and $t'' \in A^{(j)}$.
\end{proof}

\subsection{Solving \subsetsum} \label{sec:alg_subsetsum_an}
Based on Lemma~\ref{lem:feasibility_sets_addition}, we can now compute the feasible target values in the relevant interval $I^{(k)}$, with $t \in I^{(k)}$ using a binary exponentiation argument and therefore find a solution to \subsetsum.
\begin{algorithm} \caption{Solving \subsetsum} \label{alg:subsetsum}
\begin{enumerate}
    \item Use the algorithm of Bringmann~\cite{bringmann2017subset_sum_near_linear} to compute $A^{(1)}$.
    \item Compute sets $A^{(2^i)}$ iteratively by
    \begin{align*}
        A^{(2^i)} = \left(A^{(2^{i-1})} + A^{(2^{i-1})} + S \right) \cap I^{(2^i)}
    \end{align*}
    for all $i \leq R$, where $R$ is defined by the smallest index with $|A^{(2^R)}| = a_{n}$ or $2^R \geq k$. If $|A^{(2^R)}| = a_{n}$ and $2^R \leq k$ then return that $t$ is feasible.
    \item Starting with $A^{(1)}$, compute $A^{(k)}$ iteratively by
        \begin{align*}
            {A^{(\ell+2^i)} = (A^{(\ell)} + A^{(2^i)} + S) \cap I^{(\ell+2^i)}}
        \end{align*}
    for every index $i \in Bin(k)$, where $Bin(k)$ is the set of indices with a $1$ in the binary encoding of $k$, i.e. $\sum_{i \in Bin(k)} 2^i = k$.
    \item If $t \in A^{(k)}$ return that $t$ is feasible; otherwise return that $t$ is not a feasible target value.
\end{enumerate}
\end{algorithm}

\begin{theorem} \label{thm:subsetsum_runtime}
    \subsetsum~can be solved in time
    \begin{align*}
        O \Big(a_n \log a_n \cdot \log \big(\min \{\frac{t}{a_n} ; \frac{F}{a_n} \} \big) \Big) = O(a_n \log a_n \cdot \log (\frac{a_n}{n})).
    \end{align*}
\end{theorem}
\begin{proof}
    \textbf{Correctness}\\
To show the correctness of Algorithm~\ref{alg:subsetsum} one has to prove that the sets $A^{(i)}$ are computed correctly, i.e. a target value $t \in I^{(j)}$ is feasible if and only if $t \in A^{(j)}$.

The algorithm of Bringmann~\cite{bringmann2017subset_sum_near_linear} that is used in step (1) computes all feasible target values $t' \leq t$ for a given $t$ in time $O(t \log t)$. By setting $t= a_n$, we obtain the feasibility set $A^{(1)}$.

The correctness of the other feasibility sets $A^{(i)}$ for $i>1$ that are computed in step (2) and (3) of the algorithm follows directly by Lemma~\ref{lem:feasibility_sets_addition}.

\textbf{Running Time}\\
The running time of the algorithm depends on how many sumset computation are performed. As each sumset is computed for intervals $I^{(i)}$ of the same length $a_n$, the running time for each sumset computation is bounded by $O(a_n \log a_n)$.

In step (2) and (3) of the algorithm a total of $O(R) = O(\min \{ \log k, \log \frac{F}{a_n} \}) = O(\min \{ \log \frac{t}{a_n}, \log \frac{F}{a_n} \}) =  O(\frac{F}{a_n})$ are being performed. Summing up, we obtain a running time of 
\begin{align*}
    O \Big(a_n \log a_n \cdot \log \big(\min \{\frac{t}{a_n} ; \frac{F}{a_n} \} \big) \Big)
\end{align*}
for the algorithm.
Using that $F$ is bounded by $O(\frac{a_{n}^2}{n})$ \cite{erdos1972frobenius_bound}, we obtain a running time of $O(a_n \log a_n \cdot \log (\frac{a_n}{n}))$.
\end{proof}

\subsection{Solving the \frobenius~Problem} \label{sec:alg_frobenius}
Having Algorithm~\ref{alg:subsetsum} at hand, we can easily solve the \frobenius~problem. We use a binary search to determine the largest index $f$ with $|A^{(f)}| < n$ by applying Algorithm~\ref{alg:subsetsum} repeatedly to determine the sets $A^{(i)}$ and their respective cardinality. Recall that $|A^{(1)}| \leq |A^{(2)}| \leq \ldots$ and hence $|A^{(f+1)}| = n$ implies that all target values $t > f a_n$ are feasible. Therefore, the largest infeasible target value (i.e. the Frobenius number) has to be contained in $A^{(f)}$.
\begin{algorithm} \caption{Solving \frobenius} \label{alg:frobenius}
\begin{enumerate}
    \item Compute the sets $A^{(2^0)}, A^{(2^1)}, \ldots , A^{(2^R)}$ as in steps (1)--(3) of Algorithm~\ref{alg:subsetsum}, where $R$ is the smallest number with $|A^{(2^R)}| = n$.
    \item Do a binary search in the interval $[2^{R-1},2^R)$ to determine the largest index $f \in [2^{R-1},2^R)$ with $|A^{(f)}| < n$.
    \item Return the Frobenius number $F = f \cdot a_n + r$, where $r = \max \big\{(I^{(f)} \setminus A^{(f)}) \big\}$.
\end{enumerate}
\end{algorithm}
The naive approach of using Algorithm~\ref{alg:subsetsum} to determine the sets $A^{(i)}$ leads to a running time of $O(a_n \log^3 a_n)$ as the algorithm has to be called at most $\log 2^{R-1} = O(\log a_n)$ many times. Surprisingly, this running time can be improved by computing the sets $A^{(i)}$ that are required in the binary search in a more efficient way.
The following algorithm computes $f$ without repeatedly using Algorithm~\ref{alg:subsetsum} and instead relies on the precomputed feasibility sets $A^{(2^i)}$. This improves the running time of the algorithm by a logarithmic factor.
\begin{algorithm}\caption{Efficient Binary Search}
\label{alg:eff_bin_search}
\begin{algorithmic}
    \State{$L := 2^{R-1}$} \Comment{Set $L$ as the left index of the interval of the binary search}
    \For{$i=R-2$ \textbf{downto} $0$}{
        \State{$A^{(L+2^{i})} = (A^{(L)} + A^{(2^i)} + S) \cap I^{(L+2^{i})}$}
        \Comment{The middle element $L+2^i$ of the interval}
        \If {$|A^{(L+2^{i})}| < a_n$} 
            \State $L := L + 2^i$
        \EndIf
    }
    \EndFor
    \State \textbf{return }$L$
\end{algorithmic}
\end{algorithm}
Note that Algorithm~\ref{alg:eff_bin_search} simulates a classical binary search with $L$ being the leftmost index of the binary search interval and $L+2^{R-1-i}$ being the rightmost index of the binary search interval after $i$ iterations. Depending on the cardinality of the middle element $A^{(L + 2^{R-2-i})}$ the binary search continues in the left interval $[L,L + 2^{R-2-i}]$ or in the right interval $[L+2^{R-2-i}, L+2^{R-1-i}]$.

Using this algorithm in step~(2) of Algorithm~\ref{alg:frobenius} leads to the following statement.
\begin{theorem}
    The \frobenius~problem can be solved in time
    \begin{align*}
        O(a_n \log a_n \cdot \log \frac{F}{a_n}) = O(a_n \log a_n \cdot \log \frac{a_n}{n}).
    \end{align*}
\end{theorem}
\begin{proof}
    \textbf{Correctness}\\
    The correctness of the algorithm follows directly from the fact that the feasibility sets $A^{(i)}$ are computed correctly and that the binary search computes the largest index $f$ with $|A^{(f)}|<n$.
    
    \textbf{Running time}\\
    In step (1) of the algorithm, the feasibility sets $A^{(2^0)}, A^{(2^1)}, \ldots , A^{(2^R)}$ are computed which requires a running time of $O(a_n \log a_n \cdot \log (\frac{F}{a_n}))$ (see Theorem~\ref{thm:subsetsum_runtime}).
    
    In the binary search in Algorithm~\ref{alg:eff_bin_search}, a total of $O(R) = O(\log \frac{F}{a_n})$ sumset computations are performed which requires a running time of $O(a_n \log a_n \cdot \log \frac{F}{a_n})$.
    
    Step (3) of the algorithm requires only $O(a_n)$ time and hence the total running time of the algorithm can be bounded by
    \begin{align*}
        O(a_n \log a_n \cdot \log (\frac{F}{a_n})).
    \end{align*}
\end{proof}

\subsection{Solving \allvalues}
The key ingredient to efficiently solving \allvalues~is by a trade-off argument of two algorithms. The algorithm of Böcker and Lipt{\'a}k~\cite{bocker2007_frobenius_algorithm} performs well in the case that $n$ is small as it has a running time of $O(n a_0)$. The other algorithm computes the feasibility sets $A^{(i)}$ stepwise for $i \leq f$ (i.e. it computes all feasible target values $\leq F$). In the case that $n$ is large, the number of items in $A^{(1)}$ is large to begin with and we can use a bound on $F$ to bound the number of computed feasibility sets .
\begin{algorithm} \caption{Solving \allvalues}
\label{alg:allvalues_an}
\begin{enumerate}
    \item If $n < (a_n \log a_n)^{1/2}$ use the algorithm~\cite{bocker2007_frobenius_algorithm} to compute \allvalues. Otherwise proceed as follows:
    \item Use the algorithm of \cite{bringmann2017subset_sum_near_linear} to compute $A^{(1)}$.
    \item While $|A^{(i)}| < n$ do
    \begin{itemize}
        \item Set $a[i] = \min \Big\{ \{a[i] \} \cup \{  a \in A^{(i)} \mid a \equiv i \bmod a_0  \} \Big\}$
        \item Compute $A^{(i+1)}$ by\begin{align*}
        A^{(i+1)} = (A^{(i)} + S) \cap I^{(i+1)}.
        \end{align*}
        \item Set $i := i+1$.
    \end{itemize}
    \item Return sequence $a$.
\end{enumerate}
\end{algorithm}
We obtain the following statement regarding the running time of Algorithm~\ref{alg:allvalues_an}:
\begin{theorem}
    \allvalues~can be solved in time
    \begin{align*}
        O(a_{n}^{3/2} \log^{1/2}a_n).
    \end{align*}
\end{theorem}
\begin{proof}
    \textbf{Case: $n < (a_n \log a_n)^{1/2}$}\\
    In this case the algorithm of Böcker and Lipt{\'a}k~\cite{bocker2007_frobenius_algorithm} is used with a running time of 
    \begin{align*}
        O(n a_0) = O(a_{n}^{3/2} \log^{1/2}a_n).
    \end{align*}
    \textbf{Case: $n \geq (a_n \log a_n)^{1/2}$}\\ 
    In this case the sets $A^{(1)}, \ldots , A^{(1)}, A^{(f+1)}$ are computed with $|A^{(f)}| < n$ and $|A^{(f+1)}| = n$, which requires a running time of $O(f \cdot a_n \log a_n)$. By definition we know that $F \in A^{(f)}$ and hence $f \leq F / a_n$. Using the bound by Erdös and Graham~\cite{erdos1972frobenius_bound} yields
    \begin{align*}
        O(f \cdot a_n \log a_n) = O(\frac{F}{a_n} \cdot a_n \log a_n) = O(\frac{a_{n}^2 \log a_n}{n} ) = O(a_{n}^{3/2} \log^{1/2}a_n).
    \end{align*}
\end{proof}

\subsection{Hardness of Parameterization by $a_n$}
In a recent result by Abboud et al.~\cite{abboud_bringmann_hardness_subset_sum_bringmann2019} it was shown that under the SETH there exists no algorithm for \textsc{bounded subset sum} with a sublinear running time in $t$. Even more, the theorem also excludes the existence of an algorithm with a running time for example of the form $t^{1-\epsilon} \cdot 2^{\sqrt{n}}$. In their reduction they are not making use of the upper bounds of the variables which implies that the same hardness also holds for \subsetsum. This hardness statement for \subsetsum~was made precise in~\cite{Jansen_rohwedder_IP_convolution}.
\begin{theorem}[\cite{abboud_bringmann_hardness_subset_sum_bringmann2019},\cite{Jansen_rohwedder_IP_convolution}] \label{thm:hardness_subsetsum_an}
    Assuming SETH, then for any $\epsilon > 0$ there exists a $\delta > 0$ such that there exists no algorithm for \subsetsum~with a running time of 
    \begin{align*}
        O(t^{1-\epsilon}\cdot 2^{\delta n}).
    \end{align*}
\end{theorem}
Note that we can in general assume $t \geq a_n$ as items with a size $>t$ can be removed from the instance. This implies that the same hardness statement can be made for the parameter $a_n$ instead of $t$, which shows that Algorithm~\ref{alg:subsetsum} is (apart from logarithmic factors) nearly optimal under the assumption of the SETH.
In the following, we will show that we can derive the same hardness statement for \frobenius~which shows that Algorithm~\ref{alg:frobenius} is nearly optimal.

In a classical paper by Ramirez-Alfonsin~\cite{ramirez1996complexity_frobenius} it was shown that the \frobenius~problem is NP-hard. The author gave a computational reduction from \subsetsum~which 
showed that \subsetsum~can be solved for a given target value $t$ by solving only at most four instances of \frobenius. In those instances the following item sizes were used:
\begin{align} \label{eq:item_construction}
    &\bar{a}_i = 2 a_i \qquad \qquad \text{ for } 0 \leq i \leq n \notag \\ 
    &\bar{a}_{n+1} = 2 F(a_1, \ldots , a_n) + 1\\
    & \bar{a}_{n+2} = F(\bar{a}_0, \ldots \bar{a}_n,\bar{a}_{n+1})-2t \notag
\end{align}
The central theorem that is being used in the reduction is the following:
\begin{theorem}[Ramirez-Alfonsin~\cite{ramirez1996complexity_frobenius}] \label{thm:reduction_frobenius}
    For given $t \in \ZZgeq$ with $t \leq F(a_0, \ldots , a_n)$ there exists a solution of (\ref{IP1}) if and only if 
    \begin{align*}
        F(\bar{a}_0, \ldots \bar{a}_{n},\bar{a}_{n+1}, \bar{a}_{n+2} ) < F(\bar{a}_0, \ldots \bar{a}_{n},\bar{a}_{n+1}).
    \end{align*}
\end{theorem}
Based on the above theorem one can solve \subsetsum~by first computing whether $t > F(a_1, \ldots , a_n)$, in which case $t$ is feasible and in the case that $t \leq F(a_1, \ldots , a_n)$ the algorithm computes $F(\bar{a}_0, \ldots \bar{a}_{n},\bar{a}_{n+1}, \bar{a}_{n+2} )$ and $F(\bar{a}_0, \ldots \bar{a}_{n},\bar{a}_{n+1})$ in order to decide by Theorem~\ref{thm:reduction_frobenius} if $t$ is feasible.

The reason that we can not directly use this reduction in combination with Theorem~\ref{thm:hardness_subsetsum_an} is that the item sizes defined in~(\ref{eq:item_construction}) are potentially larger than $a_n$ or $t$. This is because $F(a_0, \ldots , a_n)$ can only be bounded by a quadratic term in $a_n$. In the following theorem we show how to avoid this quadratic dependency. By adding only $O(\log t)$ many additional items we can show that the Frobenius number of the new instance can be bounded by $O(t \log t)$.
\begin{theorem} \label{thm:hardness_frobenius_an}
    Assuming SETH, for any $\epsilon$ there exists a $\delta>0$ such that \frobenius~can not be solved in time
    \begin{align*}
        O(a_{n}^{1-\epsilon} \cdot 2^{\delta n}).
    \end{align*}
\end{theorem}
\begin{proof}
    Assume that there exists an algorithm $\mathcal{A}$ for \frobenius~. Then we can use it to develop an algorithm $\mathcal{B}$ for \subsetsum~solving an instance $\mathcal{I}$ with given target value $t$ and item sizes $a_0, \ldots , a_n$. Without loss of generality, we thereby assume that all item sizes are less than $t$, i.e. $a_n \leq t$. Define a new instance $\mathcal{I}'$ of \subsetsum~with additional item sizes
    \begin{align*}
        a_{n+i} = R a_0+(2^i \bmod a_0) \qquad \text{ for }0 \leq i \leq \log a_0,
    \end{align*}
    where $R \in \ZZ$ is the smallest value with $R a_0 > t$.
    The instance $\mathcal{I}'$ has $\log a_0 \leq \log t$ additional items and each item is bounded in size by $R a_0 + \log a_0 \leq 2t + \log a_0 \leq 3t$.
    Since the newly defined item sizes exceed $t$, we know that $t$ is a feasible target value for instance $\mathcal{I}$ if and only if $t$ is a feasible target value for the instance $\mathcal{I}'$.
    
    \textbf{Claim:} The Frobenius number of instance $\mathcal{I}'$ is bounded by \begin{align*}
        F(a_0, \ldots , a_{n'}) < O(t \log t).
    \end{align*}
    Consider a target value $t' \in \ZZ$ with $t' > 3 t \log t$. We will see that $t'$ can be written by a subset of the items $a_{n+1}, \ldots , a_{n'}$ and multiples of the item $a_0$.
    Let $t'' \equiv t' \bmod a_0$ and let $I$ be the set of indices with a $1$ in the binary encoding of $t''$. Then 
    $t'' \equiv t'  \equiv \sum_{i \in I} 2^i \equiv \sum_{i \in I} a_{t+i} \bmod a_0$.
    Since each $a_{t+i} \leq 3t$ we know that 
    \begin{align*}
        \sum_{i \in I} a_{t+i} \leq 3t \log t \leq t'
    \end{align*}
    and hence 
    \begin{align*}
        t' = K a_0 + \sum_{i \in I} a_{t+i}
    \end{align*}
    for some $K \in \ZZgeq$ which proves the claim.
    
    Knowing that the Frobenius number of instances $\mathcal{I}'$ is bounded by $O(t \log t)$, we can use the reduction of \cite{ramirez1996complexity_frobenius} to show the SETH based hardness for \frobenius. Algorithm~\ref{alg_reduction_frobenius_an} solves \subsetsum~by using \frobenius~as a subroutine.
    \begin{algorithm} \caption{Solving \subsetsum~\cite{ramirez1996complexity_frobenius}} \label{alg_reduction_frobenius_an}
    \begin{enumerate}
        \item If $t > F(a_1, \ldots , a_n)$ return that $t$ is feasible.
        \item Compute item sizes $\bar{a}_0, \ldots , \bar{a}_{n'+2}$.
        \item Return that $t$ is feasible if 
        \begin{align*}
            F(\bar{a}_0, \ldots \bar{a}_{n'},\bar{a}_{n'+1}, \bar{a}_{n'+2} ) < F(\bar{a}_0, \ldots \bar{a}_{n'},\bar{a}_{n'+1}).
        \end{align*}
        Otherwise return that $t$ is not feasible.
    \end{enumerate}
    \end{algorithm}
    Since $F(a_0, \ldots, a_{n'})$ is bounded by $O(t \log t)$, by definition of the item sizes in~(\ref{eq:item_construction}) the item sizes $\bar{a}_0, \ldots , \bar{a}_{n'+2}$ are also bounded by $O(t \log t)$. Note that according to \cite{ramirez1996complexity_frobenius} we have $F(a_0,\ldots, a_{n'}) = 4 F(\bar{a}_0, \ldots , \bar{a}_{n'+1})$.

    By Theorem~\ref{thm:hardness_subsetsum_an}, we know that for an arbitrary $\epsilon>0$ there exists a $\delta>0$ such that there is no algorithm with a running time $O(t^{1- \epsilon} \cdot 2^{\delta n})$ for \subsetsum. Here, we can assume that $\delta< \epsilon$ as the non-existence of an algorithm with a running time of $O(t^{1- \epsilon} \cdot 2^{\delta n})$ for $\delta \geq \epsilon$ excludes the existence of an algorithm with improved running time of $O(t^{1- \epsilon} \cdot 2^{\epsilon n})$.
    
    Assume now that there exists an algorithm with a running time of $O(a_{n}^{1-2 \epsilon}\cdot 2^{\delta n})$ for \frobenius. Then by the above construction we know that there also exists an algorithm for \subsetsum~with a running time of 
    \begin{align*}
        O((t \log t)^{1-2 \epsilon} \cdot 2^{\delta (n + \log t)}) \\
        = O((t \log t)^{1-2 \epsilon} \cdot t^{\delta}\cdot 2^{\delta n}) \\
        = O((t \log t)^{1-2 \epsilon + \delta} \cdot 2^{\delta n}) \\
        \overset{\delta < \epsilon}{=} O(t^{1- \epsilon} \cdot 2^{\delta n}).
    \end{align*}
    However, this contradicts the assumption that no such algorithms exists for \subsetsum.
\end{proof}

\section{Conclusion}
We prove nearly matching upper and lower bounds for \subsetsum,\frobenius~and \allvalues~parameterized by $a_0$ and $a_n$ respectively. The only major gap left in the open is regarding the complexity of \allvalues~parameterized by $a_n$. To us it seems that an algorithm beating the running time of $O(a_{n}^{3/2})$ is possible but very sophisticated techniques are required, like an improvement to \minconv~for the case that the entries of the sequence are bounded or sophisticated structural properties to the solution space are needed. An algorithm with a running time of $O(M^{1-\epsilon} a_n)$ for \minconv~with entries bounded by $M$ for example would improve upon the complexity of \allvalues. A hardness result showing that \allvalues~can not be solved in linear time would also be very interesting as this would imply that the complexity of a parameterization of $a_0$ and $a_n$ of the considered problems would behave differently.

Furthermore, it would also be interesting to improve upon the logarithmic factors for the algorithms parameterized by $a_n$ or to show that this is not possible under certain hardness assumptions. Note that a running time of $O(a_n \log a_n)$ is already necessary to apply a single FFT to an interval $I$ of length $a_n$. 

\bibliographystyle{abbrv}
\bibliography{library}

\end{document}